\documentclass[12pt]{article}
\usepackage{a4wide}
\usepackage{amsmath}
\usepackage{amsfonts}
\usepackage{amssymb}
\usepackage{graphicx}
\usepackage{fourier}
\usepackage{dsfont}
\usepackage{hyperref}
\usepackage{enumerate}
\usepackage{esint}
\usepackage{bm}
\usepackage{color}

\newcommand{\version}{March 15, 2016}

\usepackage{amsthm,amsfonts,amsmath, amscd}
\usepackage{amssymb,amsmath, dsfont}
\usepackage{graphicx}

\swapnumbers
\theoremstyle{plain}

\newtheorem{theorem}{Theorem}[section]

\theoremstyle{definition}
\newtheorem{definition}[theorem]{Definition}

\newtheorem{prop}[theorem]{Proposition}

\theoremstyle{remark}
\newtheorem{remark}[theorem]{Remark}                                
\newcommand{\upchi}{\raise1pt\hbox{$\chi$}}

\newcommand{\F}{{\mathcal{F}}}

\newcommand{\dd}{{\, \rm d}}
\newcommand{\tr}{{\rm Tr}}

\renewcommand{\|}{{\Vert}}

\numberwithin{equation}{section}
\def\dd{{\rm d}}
\def\M{\mathcal{M}}
\def\Q{\mathcal{Q}}
\def\N{\mathbb{N}}

\newcommand{\norm}[1]{\lVert#1\rVert}
\newcommand{\Norm}[1]{\left\lVert#1\right\rVert}
\newcommand{\abs}[1]{\left\lvert#1\right\rvert}
\newcommand{\pa}[1]{\left( #1 \right)}
\newcommand{\br}[1]{\left\lbrace #1\right\rbrace}

\newcommand{\R}{\mathbb{R}}
\newcommand{\kac}{\mathbb{S}^{N-1}\pa{\sqrt{N}}}

\numberwithin{equation}{section}
\def\S{\mathcal{S}}
\def\L{\mathcal{L}}
\def\ncht{\left(\begin{matrix} N\cr 2\cr \end{matrix}\right)}

\begin{document}

\def\tr{{\rm Tr}}

\title{Entropy production inequalities for the Kac Walk}
\author{\vspace{5pt} Eric A. Carlen$^1$, Maria C. Carvalho$^{2}$  and Amit Einav$^{3}$ \\
\vspace{5pt}\small{$1.$ Department of Mathematics, Hill Center,}\\[-6pt]
\small{Rutgers University,
110 Frelinghuysen Road
Piscataway NJ 08854-8019 USA}\\
\vspace{5pt}\small{$2.$ CMAF-CIO, University of Lisbon, P 1749-016 Lisbon, Portugal}\\
\vspace{5pt}\small{$3.$ Departments of Pure Mathematics and Mathematical Statisticsl,} \\[-6pt]
\small{University of Cambridge, Wilberforce Road, Cambridge, CB3 0WB, UK}}

\date{\version}
\maketitle 
\footnotetext                                                                         
[1]{Work partially
supported by U.S. National Science Foundation
grant DMS 1501007   }

\footnotetext
[2]{Work partially
supported by supported by Funda\c c\~ao para a Ci\^encia e Tecnologia
(PTDC/MAT/100983/2008,  PEst-OE/MAT/UI0209/2013, UID/MAT/04561/2013) }

\footnotetext                                                                         
[3]{Work supported by EPSRC grant EP/L002302/1.\\ 
\copyright\, 2017 by the authors. This paper may be
reproduced, in its
entirety, for non-commercial purposes.}

\begin{abstract}
Mark Kac introduced what is now called 'the Kac Walk' with the aim of investigating the spatially 
homogeneous Boltzmann equation by probabilistic means. 
Much recent work, discussed below,  on Kac's program has run in the other direction:  using recent results on the Boltzmann equation, 
or its one-dimensional analog, the non-linear Kac-Boltzmann equation,  to prove results for the Kac Walk.
Here we  investigate new functional inequalities for the Kac Walk pertaining to entropy production, and introduce a new form of 
`chaoticity'. We then show how these entropy production inequalities imply entropy production inequalities for the Kac-Boltzmann equation. 
This results validate Kac's program for proving results on the non-linear Boltzmann equation via analysis of the Kac Walk, and they 
constitute a partial solution to the  `Almost' Cercignani Conjecture on the sphere. 
\end{abstract}

\medskip
\leftline{\footnotesize{\qquad Mathematics subject
classification numbers: 81V99, 82B10, 94A17}}
\leftline{\footnotesize{\qquad Key Words: Kac Walk, Chaoticity,  Entropy Production, Cercignani Conjecture}}

\section{Introduction} \label{intro}  
The Kac Walk is a Markov jump process for an $N$ particle model of a gas interacting through binary collisions 
between molecules. At random times arriving in a Poisson stream, 
pairs of indistinguishable particles, with one clock for each pair, undergo an energy conserving collision in 
which their velocities are rotated at a random angle, again, chosen uniformly. In a more  general version, the Poisson clocks 
governing the collision times for pairs of particles could run at rates that are related to the 
energy of these  pairs of  particles (\cite{CCL14,V}), making 
such a collision model  more physically realistic. In the original model by Kac, \cite{K56}, these rates were uniform.
More precisely, the Kac Walk is a continuous time Markov jump process whose state space is $\kac$, 
the sphere of radius $\sqrt{N}$ in $\R^N$. Let  $\bm{v}= (v_1,\dots,v_N)$ denote a generic element of the state-space. 
The generator $\L_{N,\gamma}$, acting on continuous functions $F$ on $\kac$   is given by 
\begin{equation}\label{masgen}
\L_{N,\gamma}F ( \bm{v}) =  -{N}{\ncht}^{-1}\sum_{i<j} \pa{1+v_i^2+v_j^2}^{\gamma}\
 \frac{1}{2\pi}\int_{-\pi}^\pi \pa{F ( \bm{v}) - F(R_{i,j,\theta} \bm{v})} \dd \theta
\end{equation}
where
\begin{equation*}
(R_{i,j,\theta} \bm{v})_k = \begin{cases}v_i(\theta)=v_i \cos \theta + v_j \sin \theta & k =
 i\\   v_j(\theta)=-v_i \sin \theta+v_j \cos \theta & k = j\\   v_k & k\neq i,j\end{cases}\ ,
\end{equation*}
and $\gamma\in [0,1]$ is the parameter that measures the relation of the Poisson clocks and the energy of the colliding 
pair of particles. Let $d\sigma_N$ denote the uniform probability measure on $\kac$.
It is the unique invariant measure for this process, which is ergodic and reversible. Therefore,
the {\em {Kac Master equation}}, 
\begin{equation}\label{kme}
\frac{\partial}{\partial t} F(\bm{v},t) = \L_{N,\gamma}F(\bm{v},t)\ ,
\end{equation}
is the forward Kolmogorov equation describing the evolution of  the law  $F(\bm{v},t)d\sigma_N$ of the state under the process (assuming that the law of the initial state is absolutely continuous with respect to $d\sigma_N$.) 
Kac devised his model  to give a probabilistic framework  for the description of  ``an average $N$ particle gas''  from which he could deduce,
 in the limit $N\to\infty$,   the evolution of the single particle marginals of solutions of the Kac Master equation (\ref{kme}) with `chaotic' initial data,
by a one dimensional Boltzmann-like equation
\begin{equation}\label{kbeA}
\frac{{\partial}}{\partial t} f(v,t) =  \Q_\gamma f(v,t)
\end{equation}
where the non-linear operator $\Q_\gamma$ is given in (\ref{kbew}) below.  (In \cite{K56}, Kac only consider the case $\gamma =0$, but more recent work has extended his results to other values of $\gamma$ as discussed below.)

Kac proposed that this rigorous connection between the Master equation (\ref{kme}) and the non-linear Kac-Boltzmann 
equation (\ref{kbeA})  could be exploited to prove results on the non-linear evolution equation (\ref{kbeA}) via 
analysis of the Kac Walk. In particular, he was interested in rates of equilibration. Much recent work on the Kac program  
has gone in the opposite direction. While the rigorous mathematical investigation non-linear kinetic equations such as 
(\ref{kbeA}) was in a primitive state in 
1956 when Kac made his proposal, with most of what was known contained in a paper \cite{C32} of Carleman, 
it has advanced considerably since that time.  The recent paper \cite{MM} of Mischler and Mouhot entitled
{\em On Kac's program in kinetic theory} uses the analytic advances in the understanding of the non-linear 
Boltzmann equation to obtain deep results on the behaviour of the Kac Walk. 

In this work, we validate Kac's original vision for his program by giving direct proofs of some new functional 
inequalities for the Kac Walk from which we deduce bounds on the rates of relaxation to equilibrium for solutions of 
(\ref{kbeA}), just as Kac had proposed. The results we obtain for the Kac Walk themselves are new and interesting, 
and can be viewed as a partial positive resolution of the `Almost' Cercignani Conjecture for the Kac Walk, as 
discussed in \cite{V}. We shall return to this later, but before describing our results in more detail, it is necessary to 
explain Kac's notion of {\em chaos}, and how it provides the bridge between the Kac Master equation and the Kac-Boltzmann equation.

\subsection{The Connection with the Boltzmann Equation}

Let $F_N(\bm{v},t)$ be  a solution of the $N$-particle Kac master equation with a permutation 
symmetric initial data $F_N(\bm{v},0)$ (representing the indistinguishability of the particles).  For $N\geq k$, $\Pi_k\pa{F_N}(v_1,\dots,v_k,t)$ denote the 
$k$-particle marginal of  $F_N({\bf{v}},t)$.
That is, for continuous and bounded functions $\varphi$ in $\R^k$,
$$
\int_{\R^k} \varphi(v_1,\dots v_k) \Pi_k\pa{F_N}(v_1,\dots,v_k,t){\rm d}v_1\cdots {\rm d}v_k = 
\int_{\kac} \varphi(v_1,\dots,v_k)F_N({\bf{v}},t){\rm d}\sigma_N\ .
$$
Kac used the  permutation symmetry of the initial data and of the generator $\L_{N,\gamma}$ to show that 
$$\frac{\partial}{\partial t}\Pi_1\pa{F_N}(v_1) = \frac{1}{\pi}\int_{0}^{2\pi}\int_{\R} 
\pa{\Pi_2\pa{F_N}\pa{v_1(\theta),v_2(\theta)}-\Pi_2\pa{F_N}(v_1,v_2)}  dv_2 d\theta\ .$$
Kac noticed that {\em if} 
\begin{equation}\label{cond}
\lim_{N\to\infty} \Pi_1\pa{F_N}(v_1,t) = f(v_1,t)\qquad{\rm and}\qquad 
\lim_{N\to\infty} \Pi_2\pa{F_N}(v_1,v_2,t) =f(v_1,t)f(v_2,t) 
\end{equation}
for all $t\geq 0$, the limiting single particle marginal $f(v,t)$ solves (\ref{kbeA})
with 
\begin{equation}\label{kbew}
\Q_\gamma f(v)  =  \frac{1}{\pi}\int_{-\pi}^{\pi} \pa{1+ (v^2+w^2)}^{\gamma} [f(v \cos\theta  - w\sin\theta )f(v\sin\theta  
+ w\cos\theta ) - f(v)f(w)]{\rm d}\theta {\rm d}w.
\end{equation}

Though Kac only considered the case $\gamma =0$ in \cite{K56}, his arguments up to this point apply  equally well to all $\gamma\in [0,1]$. 
However, it remains to determine whether (\ref{cond}) is valid for solutions of the Kac Master Equation, and here the 
value of $\gamma$ makes a difference in the degree of difficulty of the problem. With (\ref{cond}) as motivation, 
Kac made the following definition:

\begin{definition}[Chaos]\label{simchaosdef} 
Let 
$f$ be a probability density on $\R$ with respect to Lebesgue measure such that  $f$ satisfies 
$\int_{\R} v^2f(v){\rm d}v =1$.  A sequence $\{F_N\}_{N\in \N}$, where $F_N$ is a permutation symmetric 
probability density on $\kac$, is called {\em $f$-chaotic} if for all $k\in N$, and all continuous bounded functions $\varphi$ on $\R^k$,
$$\lim_{N\to\infty} \int_{\S_N} \varphi(v_1,\dots,v_k) F_N(v){\rm d}\sigma_N =
 \int_{\R^k} \varphi(v_1,\dots,v_k) \prod_{j=1}^k f(v_j) {\rm d}v_1\dots {\rm d}v_k\ ,$$
 i.e. the $k-$th marginal $\Pi_k \pa{F_N}$ of $F_N$ converges to $f^{\otimes k}$ in the vague topology.
\end{definition}

Kac showed that in  case if $\{F_N\}_{N\in \N}$ is $f_0$-chaotic, and if $F_N({\bf v},t)$ is the solution of the Kac Master equation
(\ref{kme}) with $\gamma =0$ and initial data $F_N(v)$, then $\{F_N(\cdot t)\}_{N\in \N}$ is $f(\cdot,t)$-chaotic where
$f$ solves (\ref{kbeA}) with $\gamma = 0$.  This was the main result in \cite{K56}, and is the first example of a 
theorem on ``propagation of chaos'' a term coined by Kac in reference to this primal example. 

The open the question of proving propagation of chaos for values of $\gamma\in(0,1]$ was achieved later on by 
Sznitmann (see \cite{Sz1,Sz2}), building on earlier work of Gr\"unbaum (\cite{Gru}). These authors treated a more 
complicated collision model with three dimensional velocities and collisions that conserve energy as well as 
momentum, but the methods apply to the present equations as well, in a simpler manner. 

Much recent work on Kac's program has focused on {\em quantitative} refinements of Kac's Theorem giving rates of convergence in
(\ref{cond}) and in stronger topologies than the vague topology that sufficed for Kac's original purposes. 
This is the main focus of the work of Mischler and Mouhot \cite{MM} mentioned earlier.  
A refined notion of chaos plays a critical role in the present work. 

\subsection{Convergence to Equilibrium and the Entropy Method}

As  mentioned earlier, Kac had hoped to deduce rates of convergence to equilibrium for solutions of  \eqref{kbeA} 
from rates of convergence to equilibrium for solutions of his master equation. To do this, one needs rates estimates for 
the Master equation that are {\em independent of $N$}. 

There are a number of ways to measure rates of convergence to equilibrium for reversible  random walks, and 
one of the simplest is in terms of a spectral gap of the generator. Kac conjectured that the spectral gap of  
$\L_{N,0}$ was bounded away from zero, uniformly in $N$. 
This problem remained open until 2001, when it was solved by Janvresse \cite{J01}).  This result was made quantitative,  
extended to three dimensional collisions, and to $\gamma>0$ in a series of papers (see \cite{CCL00,CCL03,CGL}). 

However, uniformity in $N$ is not all that is needed to pass from rate bounds for the Master equation to rate bounds for the Kac-Boltzmann equation. The $L^2$ metric is ill-suited to this purpose (see \cite{V}), and another measure of the distance to equilibrium is required. 
Indeed, the $L^2$ metric is not particularly natural for the Kac-Boltzmann equation.  What is more natural, especially in the context of Boltzmann well-known $H$-Theorem on the monotonicity of the entropy for solutions of his equation, is the use of {\em relative entropy},
as suggested originally by Cercignani (\cite{CC}) in which he conjectured a strengthened form of the Boltzmann $H$-Theorem that we discuss below. 

Let  $f $ and $g$ are two probability 
densities on a measure space $(X,\mathcal{\F},\mu)$. The {\em relative entropy
of $f {\rm d}\mu$ with respect to $g{\rm d}\mu$} is the quantity $H(f| g) =
 \int_Xf  [\ln f - \ln g]{\rm d}\mu$. 
 Pinsker's inequality \cite{Cis,Kul,Pin} says that 
\begin{equation}\label{pinsker}
H(f|g)  \geq \frac12 \left(\int_X|f-g|{\rm d}\mu\right)^2\,
\end{equation}
Thus, while $H(f|g)$ is not itself a metric, it does control the $L^1$ distance between $f$ and $g$. 

The equilibrium solutions of the Kac-Boltzmann equation (\ref{kbeA}) are the {\em centred Maxwellian densities} 
$M_T(v) = (2\pi T)^{-1/2} e^{-v^2/2T}$. The equilibrium to which the solution with initial data $f_0$ tends is the one with
$T = \int_\R v^2 f_0(v){\rm d}v$ (see \cite{K56}).  Consider a solution $f$ of (\ref{kbeA}) with initial data $f_0$ for which $\int_{\R}v^2 f(v){\rm d}v =1$.
Since the energy is conserved, Boltzmann's $H$ theorem implies that $H(f(\cdot,t)|M_1)$ is monotone decreasing in $t$ (\cite{CC}). 
Cercignani's conjecture for the Kac-Boltzmann equation (he actually considered the analog for $3$-dimensional velocities) was that for some constant $C>0 $, all such solutions with initial data and with $H(f_0|M_1) < \infty$ satisfy
\begin{equation}\label{Cercon}
\frac{{\rm d}}{{\rm d}t} H(f(\cdot,t)|M_1)  \leq -CH(f(\cdot,t)|M_1)\ .
\end{equation}
Pinsker's inequality would then yield $\| f(\cdot,t) - M_1\|_1 \leq [H(f_0|M_1)]^{1/2} e^{-Ct/2}$. 

Cercignani's conjecture is false for all $\gamma < 1$ \cite{BC} -- but it is true for $\gamma =1$, as shown by Villani \cite{V}, who also showed how this result could be used to prove non-exponential bounds on the rate of relaxation for other values of $\gamma$ and suitable constraints on the initial data.  Thus, the best one can hope for is not (\ref{Cercon}), but something such as 
\begin{equation}\label{Cercon1.5}
\frac{{\rm d}}{{\rm d}t} H(f(\cdot,t)|M_1)  \leq -C_\epsilon \left(H(f(\cdot,t)|M_1)\right)^{1+\epsilon}\ 
\end{equation}
for $\epsilon>0$, and this is what Villani shows to be true for suitable classes of initial data, as we discuss below. For future reference, we express (\ref{Cercon1.5}) as a functional inequality. Define ${\displaystyle D(f(\cdot,t)) :=  -\frac{{\rm d}}{{\rm d}t} H(f(\cdot,t)|M_1)}$ where
$f(v,t)$ is the solution of (\ref{kbeA}) with  $f(v,0) = f(v)$.  Then we may restate (\ref{Cercon1.5}) as
\begin{equation}\label{Cercon2}
D(f(\cdot,t))   \geq C_\epsilon \left(H(f(\cdot,t)|M_1)\right)^{1+\epsilon}\ .
\end{equation}

The main question that we address here is the following:  
Do there exist functional inequalities for the Kac Walk from which it is possible to deduce inequalities of the form (\ref{Cercon2})?

To investigate this question, let $F$ be a probability density with respect to ${\rm d}\sigma_N$ on $\kac$.  The relative entropy of $F$ with respect to the uniform density $1$ is simply 
$\int_{\kac} F\ln F {\rm d}\sigma_N$. To simplify our notation, we define 
\begin{equation}\label{relent}
H_N(F) = \int_{\kac} F\ln F {\rm d}\sigma_N \ .
\end{equation}

We are thus led to  investigate the {\em (relative) entropy dissipation} under the dynamics generated by $\L_{N,\gamma}$. This dissipation, sometimes called \emph{the entropy production} is the non-negative quantity $D_{N,\gamma}(F)$ that is given by
 $$
 D_{N,\gamma}(e^{t\L_{N,\gamma}} F)  = - \frac{{\rm d}}{{\rm d}t} H_N(e^{t\L_{N,\gamma}}F)\ .
 $$
A direct computation shows that,
 $$D_{N,\gamma}(F_N)=\frac{1}{4\pi}\frac{N}{\pa{\begin{tabular}{c} $N$ \\ $2$ \end{tabular}}}\sum_{i<j}\int_{\kac}\int_0^{2\pi}\pa{1+(v_i^2+v_j^2)}^\gamma\psi\pa{F_N,F_N\circ R_{i,j,\theta}}d\sigma^Nd\theta \ .$$
where $\psi(x,y)=(x-y)\log \pa{\frac{x}{y}}$.\\
Inequalities relating $H_N(f)$ and $D_{N,\gamma}(F)$ are useful for quantifying the aforementioned rate of convergence. 
As we only connect the Kac Walk with the Kac-Boltzmann equation in the limit $N\to\infty$, we are ultimately only interested in inequalities that are {\rm uniform in $N$}. In particular, one might hope to find $C_\gamma>0$, independent of $N$, such that
$$D_{N,\gamma}(F) \geq C_\gamma H_N\pa{F}.$$
This is known as {\em{Cercignani's Conjecture for the Kac Walk}}.  A significant breakthrough in its study was done in 2003 by Villani (see \cite{V}) where he introduced the family of operators $\br{\L_{N,\gamma}}_{\gamma\in[0,1]}$ and showed that
\begin{equation}\label{epkme}
D_{N,\gamma}(F) \geq C_\gamma \frac{ H_N(F)}{N^{1-\gamma}}.\
\end{equation}
This gives a decay rate of order $e^{-CN^{\gamma-1}t}$ that, besides the case $\gamma=1$, is meaningless in the limit $N\to\infty$. 

\subsection{Main Results}

For the Kac Walk, the intuition that chaotic data with a one particle marginal $f$ behaves like  
$f^{\otimes N}$ is true in some cases, and one can show (see \cite{CCRLV} for precise statements)  that
\begin{equation}\nonumber
 \frac{ H_N(F_N)}{N} \approx H(f|M), \quad  \frac{ D_{N,\gamma}(F_N)}{N} \approx D_\gamma(f).
\end{equation}

This, together with the inequality (\ref{Cercon2}) of Villani suggests that we seek inequalities of the form
$$\frac{D_N\pa{F_N}}{N} \geq C_{\gamma,\epsilon}\pa{\frac{H_N\pa{F_N}}{N}}^{1+\epsilon}.$$
We shall prove two inequalities of this type. 
The first of these holds for a class of initial data that is propagated by the Kac Master equation.  However, in this inequality, the
constant depends weakly on $N$. However weak, this dependence prevents this inequality from being used to prove results for the Kac-Boltzmann equation. The second inequality has a constant that is independent of $N$, but the conditions on $F$ under which it is valid involve a new notion of chaos that we do not know to be propagated by the Kac Master equation. {\em Despite this},  we show that this inequality may be used to bound the rate of relaxation to equilibrium for solutions of the Kac-Boltzmann equation. 
For simplicity, from this point onwards we will mean permutation symmetric when saying that a given density function is symmetric.

\begin{definition}
Let $F_N\in P\pa{ \kac}$. We say that $F_N$ is log-scalable if there exists $C>0$, independent of $N$, such that
\begin{equation}\label{eq:log scalable}
\sup_{\bm{v}\in \kac}\abs{\log F_N\pa{\bm{v}}} \leq C N.
\end{equation}
\end{definition} 
\begin{definition}\label{def:log-power-property}
We say that a family of probability densities $\br{F_N}_{N\in\mathbb{N}}\in P\pa{\kac}$ have the log-power property of order $\beta>0$, if there exists $C>0$, independent of $N$, such that
\begin{equation}\label{eq:log-chaoticity}
\frac{1}{2\pi}\int_{0}^{2\pi}\int_{\kac} \psi_\beta\pa{F_N-F_N\circ R_{1,2,\theta}}d\sigma^N d\theta \leq C^{1+\beta}
\end{equation}
where $\psi_\beta(x,y)=\abs{x-y}\abs{\log\pa{\frac{x}{y}}}^{1+\beta}$
\end{definition}

As we explain below, the condition that  $\br{F_N}_{N\in\N}$ have the log-power property of order $\beta>0$ is a quantitative chaoticity condition, and it can be verified when 
$\br{F_N}_{N\in\N}$  is a family of {\em normalized tensor product states}, as constructed in \cite{CCRLV} from a suitable probability density $f$  on $\R$.  The conditions on $f$ that are required for this are propagated by the Kac-Boltzmann equation, and this crucial fact allows us to side-step the interesting question as to where the  log-power property of order $\beta$ might be propagated by the Kac Master Equation. 

In what follows we will use the notation $M_k(f)=\int_{\R}\abs{v}^k f(v)dv$ for any non-negative function on $\R$.   The entropy-entropy production bound with an $N$-dependent constant is:
\begin{theorem}\label{thm:mainlogscale}
Let $F_N\in P\pa{\kac}$ be symmetric and log-scalable with associated constant $C_F>0$. 
\begin{enumerate}[(i)]
\item Assume there exists $k>1$ such that
${\displaystyle M_{2k} = \sup_{N} M_{2k}\pa{\Pi_1\pa{F_N}}<\infty}$.
Then, 
\begin{equation}\label{eq:almost cerc on sphere}
\frac{D_{N,\gamma}(F_N)}{N} \geq \mathcal{C}_{k,\gamma,N}\pa{\frac{H_N(F_N)}{N}}^{1+\frac{1-\gamma}{k-1}},
\end{equation}
with $\mathcal{C}_{k,\gamma,N}=\frac{k-1}{3^{\frac{2k-\gamma k-\gamma}{k-1}}\pa{1-\gamma}}\pa{\frac{1-\gamma}{k-\gamma}}^{\frac{k-\gamma}{k-1}}\frac{1}{\pa{2C_F}^{\frac{1-\gamma}{k-1}}\pa{1+2M_{2k}}^{\frac{1-\gamma}{k-1}}} N^{{\frac{\gamma-1}{k-1}}}$.
\item AIf for  $a,\mu>0$, 
${\displaystyle M_{exp} = \sup_{N} \int_{\kac}e^{a\abs{v_1}^\mu}F_N\pa{v_1,\dots,v_N}d\sigma^N<\infty}$,
then
\begin{equation}\label{eq:almost cerc on sphere exp moment}
\frac{D_{N,\gamma}(F_N)}{N} \geq \frac{1}{6\cdot 4^{1-\gamma}\pa{\frac{2}{a}\log\pa{96C_F  \pa{\frac{4}{a \mu e}}^{\frac{2}{\mu}} e^{\frac{a}{2^{\mu/2}}}M_{exp}}+\frac{2}{a}\log\pa{\frac{N}{\pa{\frac{H_N(F_N)}{N}}}}}^{\frac{2(1-\gamma)}{\mu}}} \frac{H_N(F_N)}{N}
\end{equation}
\end{enumerate}
\end{theorem}

The entropy-entropy production bound with a  constant independent of $N$ is:

\begin{theorem}\label{thm:mainlofpower}
Let $F_N \in P\pa{\kac}$ have the log-power property of order $\beta$ with associated constant $C_F>0$. 
\begin{enumerate}[(i)]
\item If there exists $k>1+\frac{1}{\beta}$ such that
$M_{2k} = \sup_{N} M_{2k}\pa{\Pi_1\pa{F_N}}<\infty$.  Then
$$
\frac{D_{N,\gamma}(F_N) }{N} \geq \mathcal{C}_\epsilon \left(\frac{H_N(F_N)}{N} \right)^{1+\epsilon}
$$ where
${\displaystyle \epsilon = 1+\frac{(1-\gamma)(1+\beta)}{k\beta-(1+\beta)}}$
and
\begin{equation}\label{eq:almost cerc on sphere II}
\mathcal{C}_\epsilon = \frac{k\beta-(1+\beta)}{(1+\beta)(1-\gamma)}\pa{\frac{(1+\beta)(1-\gamma)}{k\beta-(1+\beta)}}^{\frac{k\beta-\gamma(1+\beta)}{k\beta-(1+\beta)}}
\frac{2^{\frac{1-\gamma}{k\beta-(1+\beta)}}}{3^{\frac{2k\beta-k\beta\gamma-\gamma(1+\beta)}{k\beta-(1+\beta)}}}
\frac{1}{C_F^{\frac{(1+\beta)(1-\gamma)}{k\beta-(1+\beta)}}\pa{1+2M_{2k}}^{\frac{\beta(1-\gamma)}{k\beta-(1+\beta)}}}\ .
\end{equation}
\item If for $a,\mu>0$
${\displaystyle M_{exp} = \sup_{N} \int_{\kac}e^{a\abs{v_1}^\mu}F_N\pa{\bm{v}}d\sigma^N<\infty}$, 
then 
\begin{equation}\label{eq:almost cerc on sphere exp moment II}
\frac{D_{N,\gamma}(F_N)}{N} 
\geq  \abs{\frac{2^{1+\mu}}{a}\log\pa{\frac{4C_F^{\frac{1+\beta}{\beta}}\pa{\frac{2^{2+\mu}(1+\beta)}{a \beta\mu e}}^{\frac{2(1+\beta)}{\beta\mu}}e^{\frac{a}{2^{\mu/2}}}M_{exp}}{\pa{\frac{H_N(F_N)}{6N}}^{\frac{1+\beta}{\beta}}}}}^{-\frac{2(1-\gamma)}{\mu}}\frac{H_N(F_N)}{N}
\end{equation}
\end{enumerate}
\end{theorem}
Because we do not know that the log-power property of order $\beta$ is propagated by the Kac-Master equation, we cannot use
Theorem ~\ref{thm:mainlogscale} to get a rate of entropic convergence for the Kac Walk that is independent of $N$. Nonetheless,
for the reasons discussed above, we can apply it to draw the following conclusion for the Kac-Boltzmann equation:

\begin{theorem}\label{thm:almost_cerc_kac_boltz}
Let $f\in P\pa{\R}$ be such that $M_2(f)=1$. Assume in addition that there exists $\beta>0$ and $k>1+1/\beta$ such that
$$M_{\max\pa{2k,k(1+\beta),4}}(f) <\infty,$$
that 
$$I(f) = \int_{\R}\frac{\pa{f^\prime(x)}^2}{f(x)}dx < \infty,$$
and that
$$f(v) \geq C e^{-\abs{v}^2} \quad\quad \forall v\in\R.$$
Then, there exists an explicit constant, $\mathcal{C}$, depending only on the parameters of the problems such that
\begin{equation}\label{eq:almost_cerc_kac_boltz}
D_\gamma(f) \geq \mathcal{C} H\pa{f|M}^{1+\frac{(1-\gamma)(1+\beta)}{k\beta-(1+\beta)}}.
\end{equation}
If the moment condition is replaced by the condition 
$M_{exp}(f)=\int_{\R}e^{a\abs{v}^\mu}f(v)dv <\infty$
for some $a,\mu>0$ then \eqref{eq:almost_cerc_kac_boltz} can be replaced by
\begin{equation}\label{eq:almost_cerc_kac_boltz_exp}
D_\gamma(f) \geq \mathcal{C}_1 H\pa{f|M}\abs{\log\pa{\frac{\mathcal{C}_2}{H\pa{f|M}}}}^{-\frac{1-\gamma}{\mu}}.
\end{equation}
Moreover, if $f(t)$ is the solution to the Kac-Boltzmann equation then the constants in \eqref{eq:almost_cerc_kac_boltz} can be chosen to be independent of time and as such a rate of convergence to equilibrium can be obtained.
\end{theorem}
Note that while the main result of Theorem \ref{thm:almost_cerc_kac_boltz} is exactly like Villani's result for the Boltzmann equation (and can probably be proved in a similar way), the method of proof we provide here not only harkens back to Kac's program and views inequalities \eqref{eq:almost_cerc_kac_boltz} and \eqref{eq:almost_cerc_kac_boltz_exp} as limit inequalities, but also uses strong connection to the field of applied probability, and in particular, to ideas of concentration of density functions on appropriate sets.

\subsection{Organization of the Paper}
The structure of the presented work is as follows: In Section \ref{sec:functionalineq} we will prove our main functional inequalities for the sphere, Theorems \ref{thm:mainlogscale} and \ref{thm:mainlofpower}, and discuss their suitability to the Kac Walk. In Section \ref{sec:g concentration} we will begin our preparation to prove Theorem \ref{thm:almost_cerc_kac_boltz} and define  new notions of concentration, and chaoticity that will play a crucial role in this work. Lastly, Section \ref{sec:vision} provides the proof of Theorem \ref{thm:almost_cerc_kac_boltz} and a short discussion on the connection between  entropic inequalities for the Kac Walk and for the Kac-Boltzmann equation.

\smallskip
\noindent{\bf Acknowledgement} We thank the anonymous referee for valuable suggestions that have improved the presentation.

\section{Functional Inequalities for the Kac Walk}\label{sec:functionalineq}
This section is dedicated to the proof of  Theorems \ref{thm:mainlogscale} and \ref{thm:mainlofpower}, which give an improvement to \eqref{epkme}. We will also discuss the implications of these Theorems to the Kac Walk.
We start by recalling a  Theorem from Villani's work, \cite{V} which is is the special case $\gamma=1$ of \eqref{epkme}:
\begin{theorem}\label{thm:Villani}
Let $F_N\in P\pa{\kac}$ such that $H_N\pa{F_N}<\infty$. Then
\begin{equation}\label{eq:Villani}
D_{N,1}\pa{F_N} \geq \frac{1}{3}H_N\pa{F_N}.
\end{equation}
\end{theorem}
With this at hand, we we proceed to prove our theorems.
\begin{proof}[Proof of Theorem \ref{thm:mainlogscale}]
$(i)$ For a given $\lambda>0$ we define $A_{\lambda}=\br{v_1,v_2\in \R | 1+v_1^2+v_2^2 > \lambda}$. Using the symmetry of $F_N$ we find that
\begin{equation}\label{eq:mainestimationformaintheorems}
\begin{gathered}
D_{N,1}(F_N)=\frac{N}{4\pi} \int_{\kac}\int_{0}^{2\pi}\pa{1+v_1^2+v_2^2}\psi\pa{F_N,F_N\circ R_{1,2,\theta}}d\sigma^N d\theta\\
\leq \lambda^{1-\gamma} D_{N,\gamma}(F_N) +\frac{N}{2\pi}\int_{\kac \cap A_{\lambda}}\int_{0}^{2\pi}\pa{1+v_1^2+v_2^2}  \pa{F_N-F_N\circ R_{1,2,\theta}}\log F_Nd\sigma^N d\theta. 
\end{gathered}
\end{equation}
Using the log-scalability we find that
\begin{equation}\label{eq:almost cerc on sphere I}
D_{N,1}\pa{F_N}\leq \lambda^{1-\gamma} D_{N,\gamma}(F_N) + 2C_F N^2\int_{\kac\cap A_{\lambda}}\pa{1+v_1^2+v_2^2} F_Nd\sigma^N.
\end{equation}
Thus, 
\begin{equation}\nonumber
\begin{gathered}
\frac{D_{N,1}(F_N)}{N} \leq \lambda^{1-\gamma} \frac{D_{N,\gamma}(F_N)}{N}+2C_F \frac{N}{\lambda^{k-1}} \int_{\kac}3^{k} \pa{\max\pa{1,v_1^2,v_2^2}}^k F_N d\sigma^N\\
\leq \lambda^{1-\gamma} \frac{D_{N,\gamma}(F_N)}{N}+2 \cdot 3^k C_F \frac{N}{\lambda^{k-1}} \pa{1+2M_{2k}}
\end{gathered}
\end{equation}
As the minimum of the function $g_{\gamma,k}(x)=a x^{1-\gamma} +b x^{1-k} $ on $(0,\infty)$, when $0\leq \gamma<1<k$, is attained at $x=\pa{(k-1)b/a(1-\gamma)}^{1/(k-\gamma)}$ and equals
$$\min_{x>0}g_{\gamma,k}(x)=\frac{k-\gamma}{1-\gamma}\pa{\frac{1-\gamma}{k-1}}^{\frac{k-1}{k-\gamma}}a^{\frac{k-1}{k-\gamma}}b^{\frac{1-\gamma}{k-\gamma}} $$
we conclude, by optimizing $\lambda$, that
$$\frac{D_{N,1}(F_N)}{N}  \leq  3^{\frac{k(1-\gamma)}{k-\gamma}}\pa{2C_F}^{\frac{1-\gamma}{k-\gamma}}N^{\frac{1-\gamma}{k-\gamma}}\pa{1+2M_\gamma}^{\frac{1-\gamma}{k-\gamma}}\frac{k-\gamma}{1-\gamma}\pa{\frac{1-\gamma}{k-1}}^{\frac{k-1}{k-\gamma}} \pa{\frac{D_N(F_N)}{N}}^{\frac{k-1}{k-\gamma}}. $$
Using \eqref{eq:Villani}, we conclude the proof of $(i)$. To show $(ii)$ we use \eqref{eq:mainestimationformaintheorems} with the fact that
$$\max_{x\geq 0}x e^{-\lambda x^{\mu}}=\frac{1}{\pa{\lambda \mu e}^{\frac{1}{\mu}}}$$
to conclude that for any $a_1,\mu_1>0$
$$\frac{D_{N,1}(F_N)}{N} \leq \lambda^{1-\gamma} \frac{D_{N,\gamma}(F_N)}{N}+2C_F N \pa{\frac{2}{a_1 \mu_1 e}}^{\frac{1}{\mu_1}}e^{-\frac{a_1\lambda^{\mu_1}}{2}} \int_{\kac}e^{a_1\pa{1+v_1^2+v_2^2}^{\mu_1}}F_N d\sigma^N.$$
Since if $\abs{v_1}\geq \abs{v_2}$ we have that
${\displaystyle e^{a_1\pa{1+v_1^2+v_2^2}^{\mu_1}} \leq e^{a_1 \pa{1+2\abs{v_1}^2}^{\mu_1}}  \leq e^{2^{\mu_1}a_1}e^{4^{\mu_1}a_1 \abs{v_1}^{2\mu_1}}}$
and symmetrically if $\abs{v_2}\geq \abs{v_1}$
${\displaystyle e^{a_1\pa{1+v_1^2+v_2^2}^{\mu_1}} \leq  e^{2^{\mu_1}a_1}e^{4^{\mu_1}a_1 \abs{v_2}^{2\mu_1}}}$.
We conclude that for $a_1=\frac{a}{2^{\mu}}$ and $\mu_1=\frac{\mu}{2}$
\begin{equation}\label{eq:almost_exp_log_scalability_I}
\frac{D_{N,1}(F_N)}{N} \leq \lambda^{1-\gamma} \frac{D_{N,\gamma}(F_N)}{N}+16C_F N \pa{\frac{4}{a \mu e}}^{\frac{2}{\mu}}e^{-\frac{a\lambda^{\frac{\mu}{2}}}{2^{1+\mu}}} e^{\frac{a}{2^{\mu/2}}}M_{exp}
\end{equation}
Using \eqref{eq:Villani} and choosing 
${\displaystyle\lambda=\pa{\frac{2^{1+\mu}}{a}\log\pa{\frac{96C_F N \pa{\frac{4}{a \mu e}}^{\frac{2}{\mu}} e^{\frac{a}{2^{\mu/2}}}M_{exp}}{\pa{\frac{H_N(F_N)}{N}}}}}^{\frac{2}{\mu}}}$
so that the second term on the right hand side of \eqref{eq:almost_exp_log_scalability_I} equals $\frac{H_N(F_N)}{6N}$ yields the desired result.
\end{proof}
\begin{proof}[Proof of Theorem \ref{thm:mainlofpower}]
The proof follows the lines of Theorem \ref{thm:mainlogscale}. \\
$(i)$ Using the same definition as those in the proof of Theorem \ref{thm:mainlogscale} we find that
$$D_{N,1}(F_N)\leq \lambda^{1-\gamma} D_{N,\gamma}(F_N) $$
$$+\frac{N}{4\pi}\int_{\kac\cap A_{\lambda}}\int_{0}^{2\pi}\pa{1+v_1^2+v_2^2} \psi\pa{F_N,F_N\circ R_{1,2,\theta}}d\sigma^N d\theta \leq \lambda^{1-\gamma} D_{N,\gamma}(F_N)$$
\begin{equation}\label{eq:almost cerc on sphere imp}
\begin{gathered}
 + \frac{N}{2}\pa{\frac{1}{2\pi}\int_{0}^{2\pi}\int_{\kac}\abs{\log\pa{F_N}-\log\pa{F_N\circ R_{1,2,\theta}}}^{1+\beta}\abs{F_N-F_N\circ R_{1,2,\theta}}d\sigma^N d\theta}^{\frac{1}{1+\beta}}\\
\pa{2\int_{\kac\cap A_{\lambda}}\pa{1+v_1^2+v_2^2}^{\frac{1+\beta}{\beta}} F_Nd\sigma^N}^{\frac{\beta}{1+\beta}}.
\end{gathered}
\end{equation}
Using the log-power property, and the fact that $\pa{1+v_1^2+v_2^2}^{\eta} \leq 3^{\eta}\pa{1+\abs{v_1}^{2\eta}+\abs{v_2}^{2\eta}}$ we find that
$$\frac{D_{N,1}(F_N)}{N} \leq \lambda^{1-\gamma} \frac{D_{N,\gamma}(F_N)}{N}+\frac{2^{\frac{\beta}{1+\beta}}3^{\frac{k\beta}{1+\beta}}C_F}{2} \frac{1}{\lambda^{\frac{k\beta}{1+\beta}-1}} \pa{1+2M_{2k}}^{\frac{\beta}{1+\beta}}.$$
Optimising over $\lambda$ yields
$$\frac{D_{N,1}(F_N)}{N}  \leq  \frac{k\beta-\gamma(1+\beta)}{(1+\beta)(1-\gamma)}\pa{\frac{(1+\beta)(1-\gamma)}{k\beta-(1+\beta)}}^{\frac{k\beta-(1+\beta)}{k\beta-\gamma(1+\beta)}}
\frac{3^{\frac{k\beta(1-\gamma)}{k\beta-\gamma(1+\beta)}}C_F^{\frac{(1+\beta)(1-\gamma)}{k\beta-\gamma(1+\beta)}}}{2^{\frac{1-\gamma}{k\beta-\gamma(1+\beta)}}}\pa{1+2M_{2k}}^{\frac{\beta(1-\gamma)}{k\beta-\gamma(1+\beta)}}
 \pa{\frac{D_{N,\gamma}(F_N)}{N}}^{\frac{k\beta-(1+\beta)}{k\beta-\gamma(1+\beta)}}. $$
from which the result follows with \eqref{eq:Villani}.\\
$(ii)$ Again, like in the proof of Theorem \ref{thm:mainlogscale} we find that the log-power property implies that for any $a_1,\mu_1>0$
$$\frac{D_{N,1}(F_N)}{N} \leq \lambda^{1-\gamma} \frac{D_{N,\gamma}(F_N)}{N}+C_F  \pa{2\pa{\frac{2(1+\beta)}{a_1 \beta \mu_1 e}}^{\frac{1+\beta}{\beta \mu_1}}e^{-\frac{a_1\lambda^{\mu_1}}{2}} \int_{\kac}e^{a_1\pa{1+v_1^2+v_2^2}^{\mu_1}}F_N d\sigma^N}^{\frac{\beta}{1+\beta}},$$
from which we get with the choice of  $a_1=\frac{a}{2^{\mu}}$ and $\mu_1=\frac{\mu}{2}$
$$\frac{D_{N,1}(F_N)}{N} \leq \lambda^{1-\gamma} \frac{D_{N,\gamma}(F_N)}{N}+4^{\frac{\beta}{1+\beta}}C_F  \pa{\frac{2^{2+\mu}(1+\beta)}{a \beta\mu e}}^{\frac{2}{\mu}}e^{-\frac{\beta a\lambda^{\frac{\mu}{2}}}{(1+\beta)2^{1+\mu}}} e^{\frac{\beta a}{(1+\beta)2^{\mu/2}}}M_{exp}^{\frac{\beta}{1+\beta}}$$
Choosing 
$$\lambda=\abs{\frac{2^{1+\mu}}{a}\log\pa{\frac{4C_F^{\frac{1+\beta}{\beta}}\pa{\frac{2^{2+\mu}(1+\beta)}{a \beta\mu e}}^{\frac{2(1+\beta)}{\beta\mu}}e^{\frac{a}{2^{\mu/2}}}M_{exp}}{\pa{\frac{H_N(F_N)}{6N}}^{\frac{1+\beta}{\beta}}}}}^{\frac{2}{\mu}}$$
so that the second term on the right is less than or equals to $\frac{H_N(F_N)}{6N}$ yields the desired result.
\end{proof}
Now that we have shown the functional inequalities, the first question we are facing is - Will any of them help gain an explicit rate of convergence to equilibrium in the Kac Walk? In order for that to be true, one will need to show that the conditions of the theorems are propagated via the flow of the master equation. This is indeed the case for inequality \eqref{eq:almost cerc on sphere}.
\begin{theorem}\label{thm:rateofconvergenceonKac}
Let $F_N\in P\pa{\kac}$ be a log-scalable function with finite first marginal moment of order $2k>2$. Then, if $F_N(t)$ is the solution to the master equation with initial datum $F_N$, $F_N(t)$ is log-scalable with the same constant as that of $F_N$, and $\M_{2k}=\sup_{N}\sup_{t\geq 0} M_{2k}\pa{\Pi_1\pa{F_N(t)}}<\infty$. As a result, there exists a constant $C>0$ that is independent in $N$ such that
\begin{equation}\label{eq:rateofconvergenceonKac}
\frac{H_N\pa{F_N(t)}}{N} \leq \pa{\pa{\frac{H_N\pa{F_N}}{N}}^{\frac{\gamma-1}{k-1}}+CN^{\frac{\gamma-1}{k-1}}t}^{-\frac{k-1}{1-\gamma}}.
\end{equation}
\end{theorem}
\begin{proof}
The boundedness of the moments is a known property of the master equation. As for the log-scalability, it is easy to see that $F_N$ is log-scalable if and only if there exists $C_F>0$ such that
$$e^{-C_F N} \leq F_N\pa{\bm{v}} \leq e^{C_F N}$$
for all $\bm{v}\in\kac$. We claim that that the solutions to the master equation propagate lower and upper bounds. Indeed, writing
$\L_{N,0}=N\pa{Q-I}$,
one notices that for any $C\in \R$, $Q(C)=C$ and if $F\geq G \geq 0$ then $Q\pa{F}\geq Q\pa{G}\geq 0$. Thus, if $F_N(0) \geq C>0$
$$F_N(t) =e^{-Nt}\sum_{n=0}^\infty \frac{N^n t^n  Q^{n}\pa{F_N(0)}}{n!} \geq  e^{-Nt}\sum_{n=0}^\infty \frac{N^n t^n  Q^{n}\pa{C}}{n!}=C.$$
A similar argument shows the propagation of a lower bound, and we conclude  the propagation of the log-scalability. Thus, using Theorem \ref{thm:mainlogscale} we conclude the existence of a constant $C>0$, independent in $N$ such that
$$\frac{D_{N,\gamma}\pa{F_N(t)}}{N} \geq \frac{C}{N^{\frac{1-\gamma}{k-1}}}\pa{\frac{H_N\pa{F_N(t)}}{N}}^{1+\frac{1-\gamma}{k-1}}.$$
Since $D_N\pa{F_N(t)}= - \frac{d}{dt}H_N\pa{F_N(t)}$, inequality \eqref{eq:rateofconvergenceonKac} is obtained.
\end{proof}

\begin{remark}\label{rem:betterthanvillanirate}
It is important to notice, as seen in  the proof of Theorem \ref{thm:rateofconvergenceonKac}, that the concept of log-scalability is mainly
for obtaining lower  bounds on $F_N$. This corresponds to the lower bound assumption on $f(v)$ in Villani's work \cite{V}. Moreover, 
 the right hand side of \eqref{eq:rateofconvergenceonKac} behaves like $N t^{\frac{k-1}{\gamma-1}}$ in those cases where $F_N$ is 'chaotic enough' so that $\frac{H_N(F_N)}{N}\approx H(f|M_1)$. While this algebraic rate of convergence is not as good as that obtained from Villani's inequality, \eqref{epkme}, for long times - it is superior to it in the relationship between $N$ and $t$! Indeed, Villani's inequality shows significant decay around $t\approx N^{1-\gamma}$, while our new result gives significant decay around $t\approx N^{\frac{1-\gamma}{k-1}}$.
\end{remark}

Concerning  the log-power property invoked Theorem \ref{thm:mainlofpower}, at this current stage we have no proof that this property propagates. Moreover, we believe that it is not the case for a fixed $\beta>0$, though this may be true if $\beta$ is allowed to change and to depend on $t$. Theorem \ref{thm:mainlofpower} is enough to understand rates of convergence to equilibrium in the Kac Walk's limit equation - The Kac Boltzmann equation.
The next section is dedicated to setting up the tools, mainly special states on Kac's sphere, that will provide the link to achieve this goal.

\section{Conditioned Tensorisation and $g$-Concentration.}\label{sec:g concentration}
In this section we introduce a special  family of densities on ${\kac}$ that will play an important role in forming a bridge between Theorem \ref{thm:mainlofpower} and Theorem \ref{thm:almost_cerc_kac_boltz}. This class of densities  has been in the forefront of the study of the Kac Walk in the past 10 years, and we single out a subclass of such densities that is particularly relevant here. More information concerning these densities can be found in  \cite{CCRLV,CE, Einav1}.

\begin{definition}\label{def: conditioned tensorisation}
Given $f\in P\pa{\R}$ with a unit second moment. We define \emph{the conditioned tensorisation of $f$} to be the probability density function $F_N\in P\pa{\kac}$
\begin{equation}\nonumber
F_N=\frac{f^{\otimes N}}{\mathcal{Z}_{N}\pa{f,\sqrt{N}}},
\end{equation}
where {\em{the normalization function of $f$}}, $\mathcal{Z}_N\pa{f,r}$, is defined as
${\displaystyle 
\mathcal{Z}_N\pa{f,r}=\int_{\mathbb{S}^{N-1}(r)}f^{\otimes N}d\sigma_r^N,
}$
with $d\sigma_r^N$ being the uniform probability measure on $\mathbb{S}^{N-1}(r)$.
\end{definition}
The conditioned tensorisation of $f$ captures the intuition that we expect chaotic states to satisfy $F_N\approx f^{\otimes N}$. Whether or not this is valid for a particular $f$ depends largely  on $\mathcal{Z}_N\pa{f,\sqrt{N}}$ which measures how concentrated $f^{\otimes N}$, the density function for the ensemble $\pa{V_1,\dots,V_N}$ where $\abs{V_i}$ are i.i.d random variables with density function $f$, on the sphere $\kac$, which represents the mean of the energy variable $\sum_{i=1}^N V_i^2$. The fact that in many cases, this concentration (expressed usually by tools of local Central Limit Theorems) is enough to understand phenomena in the Kac Walk is exploited in many recent works. We thus start our study of conditioned tensorisation with a new definition of this concentration phenomena:
\begin{definition}\label{def: g concentration on the sphere}
We say that a function $f\in P\pa{\R}$ with a unit second moment is {\em $g-$concentrated}  if the normalization function of $f$ is well defined and for any fixed $k$
\begin{equation}\label{eq: g concentration on the sphere}
\mathcal{Z}_{N-k}(f,\sqrt{u})=\frac{2}{\pa{N-k}^{\frac{1}{2}}\abs{\mathbb{S}^{N-k-1}}  u^{\frac{N-k-2}{2}}}\pa{ g\pa{u-(N-k),N-k}+\lambda_{N-k}(u)},
\end{equation}
where:
\begin{enumerate}[(i)]
\item $g$ is non-negative and bounded.
\item ${\displaystyle g(0,N)= \frac{1}{\sqrt{2\pi}} }$ for all $N$.
\item $\lim_{N\rightarrow\infty}\sup_{u}\abs{\lambda_{N}(u)}=0$.
\item For any $R>0$, ${\displaystyle \lim_{N\rightarrow\infty} \sup_{\abs{x}<R}\abs{g(x,N)-\frac{1}{\sqrt{2\pi}}}=0}$.
\end{enumerate}
\end{definition}

The following theorem in \cite{CCRLV} provides the examples of relevance here:
\begin{theorem}\label{thm: CCLLV approximation thm}
Let $f\in P\pa{\R}\cap L^p\pa{\R}$ for some $p>1$. Assume that $f$ has a unit second moment and a bounded fourth moment. Then $f$ is $g-$concentrated where 
$$ g(x,N)=\frac{e^{-\frac{x^2}{2N\Sigma^2}}}{\sqrt{2\pi}}  \quad{\rm and}\quad \Sigma^2=\int_\R v^4f(v)dv -1.$$
\end{theorem}

\begin{remark}  
A few things to notice:
\begin{itemize}
\item The proof of Theorem~\ref{thm: CCLLV approximation thm} is by a Local Central Limit Theorem. Indeed, one can easily show that
$$\mathcal{Z}_N\pa{f,\sqrt{r}}=\frac{2h^{\ast N}}{\pa{\abs{\mathbb{S}^{N-1}a}a  r^{\frac{N-2}{2}}}},$$
where $h$ is the density function associated to the random variable $V^2$.
\item  The notion of $g$-concentration can be extended to accommodate L\'evy 
Local Central Limit Theorems. The applications of these to chaos have been developed in \cite{CE}. The changes to 
Definition~\ref{def: g concentration on the sphere} required in this context are that the the initial factor of $(N-k)^{1/2}$ on the right in 
(\ref{eq: g concentration on the sphere}) would be replaced by $(N-k)^{1/\alpha}$ for some $\alpha$, and that
of $1/\sqrt{2\pi}$ in the conditions (ii) and (iv) would be replaced by another constant. 
\end{itemize}
\end{remark}
With the notion of $g-$concentration we can show that conditioned tensorisation of $f$, with appropriate $f$ are not only good candidates to the Kac Walk, in the sense that they are chaotic - they also give a formal proof to the intuition we had for the scaling of $H_N$ and $D_N$ by $N$.
\begin{definition}\label{def:extrachoas}
We say that a symmetric density probability $F_N\in P\pa{\kac}$ is $f-$entropically chaotic, for some $f\in P\pa{\R}$, if $F_N$ is $f-$chaotic and 
\begin{equation}\nonumber
\lim_{N\rightarrow\infty}\frac{H_N\pa{F_N}}{N}=H\pa{f|M}\ .
\end{equation}
We say that a symmetric density probability $F_N\in P\pa{\kac}$ is strongly $f-$entropically chaotic, for some $f\in P\pa{\R}$, if $F_N$ is $f-$entropically chaotic and 
\begin{equation}\label{strongent}
\lim_{N\rightarrow\infty}\frac{D_N\pa{F_N}}{N}=\frac{D_0(f)}{2}.
\end{equation}
\end{definition}
\begin{theorem}\label{thm:conditionedproperties}
Let $f$ be $g-$concentrated and  let $F_N$ be the conditioned tensorisation of $f$. Then 
\begin{enumerate}[(i)]
\item for any $1\leq k <N$
\begin{equation}\label{nice}\Pi_k\pa{F_N}(v_1,\dots,v_k)=\pa{\frac{N}{N-k}}^{1/2}\frac{g\pa{k-\sum_{i=1}^k v_i^2,N-k}+\lambda_{N-k}\pa{N-\sum_{i=1}^k v_i^2}}{(2\pi)^{-1/2}+\lambda_N\pa{N}}f^{\otimes k}(v_1,\dots,v_k)  \ .
\end{equation}
Consequently, $F_N$ is $f-$chaotic.
\item If in addition 
$$\Norm{f}_{L\log L}=\int_{\R}\abs{f(v)}\pa{1+\abs{v}^2}\pa{1+\abs{\log f(v)}}dv<\infty$$
then
\begin{equation}\label{eq: rescaled entropy convergence}
\lim_{N\rightarrow\infty}\frac{H_N\pa{F_N}}{N}=H\pa{f|M},
\end{equation}
i.e. $F_N$ is $f-$entropically chaotic.
\item If $\Norm{f}_{L \log L}<\infty$ then
\begin{equation}\label{eq: strongly entropically chaotic}
\lim_{N\rightarrow\infty}\frac{D_{N,\gamma}\pa{F_N}}{N}=\frac12 D_\gamma\pa{f},
\end{equation}
i.e. $F_N$ is strongly $f-$entropically chaotic. 
\end{enumerate}
\end{theorem}
\begin{proof}
These properties are mainly known under the setting of Theorem \ref{thm: CCLLV approximation thm} (or an appropriate L\'evy Central Limit Theorem), see for instance \cite{CCRLV,CE,Einav1}. The fact that the exact known concentration function $g$ in this studies is replaced by an unknown one from Definition \ref{def: g concentration on the sphere} plays no role. The important detail here is {\em{the existence of a measure of concentration}}. We will show $(iii)$ to emphasise this point. For simplicity we will only consider the case $\gamma=1$.\\
A simple geometric identity on $\kac$ shows that
\begin{equation}\label{eq: Fubini on the sphere}
\begin{gathered}
\int_{\mathbb{S}^{N-1}(r)}\phi_k F_Nd\sigma^N_r=\frac{\abs{\mathbb{S}^{N-k-1}}}{\abs{\mathbb{S}^{N-1}}}\frac{1}{r^{N-2}}\int \phi_k(v_1,\dots,v_k)\pa{r^2-\sum_{i=1}^k v_i^2}_{+}^{\frac{N-k-2}{2}}\\
\pa{\int_{\mathbb{S}^{N-k-1}\pa{\sqrt{r^2-\sum_{i=1}^k v_i^2}}}F_Nd\sigma^{N-k}_{\sqrt{r^2-\sum_{i=1}^k v_i^2}}}dv_1\dots dv_k.
\end{gathered}
\end{equation}
Using the properties of the $g-$concentration and (\ref{eq: Fubini on the sphere}) we have that
\begin{equation}\nonumber
\begin{gathered}
\frac{D_{N,1}\pa{F_N}}{N}=\frac{1}{4\pi}\frac{1}{\pa{\begin{tabular}{c} $N$ \\ $2$ \end{tabular}}} \sum_{i<j} \int_{v_i^2+v_j^2 \leq N}
\frac{\abs{\mathbb{S}^{N-3}}\pa{N-v_i^2-v_j^2}^{\frac{N-4}{2}}\mathcal{Z}_{N-2}\pa{f,\sqrt{N-v_i^2-v_j^2}}}{\abs{\mathbb{S}^{N-1}}N^{\frac{N-2}{2}}\mathcal{Z}_N\pa{f,\sqrt{N}}}\\
\pa{1+v_i^2+v_j^2}\pa{f\pa{v_i}f\pa{v_j}-f\pa{v_i(\theta)}f\pa{v_j(\theta)}}\log\pa{\frac{f\pa{v_i}f\pa{v_j}}{f\pa{v_i(\theta)}f\pa{v_j(\theta)}}}dv_idv_jd\theta.
\end{gathered}
\end{equation}
\begin{equation}\nonumber
\begin{gathered}
=\frac{1}{4\pi} \pa{\frac{N}{N-2}}^{1/2}\int_{v_1^2+v_2^2 \leq N}\frac{g\pa{2-v_1^2-v_2^2,N-2}+\lambda_{N-2}\pa{N-v_1^2-v_2^2}}{(2\pi)^{-1/2}+\lambda_N\pa{N}}\\
\pa{1+v_1^2+v_2^2}\pa{f\pa{v_1}f\pa{v_2}-f\pa{v_1(\theta)}f\pa{v_2(\theta)}}\log\pa{\frac{f\pa{v_1}f\pa{v_2}}{f\pa{v_1(\theta)}f\pa{v_2(\theta)}}}dv_1dv_2d\theta.
\end{gathered}
\end{equation}
Using the uniform convergence on compacts in the $g$-concentration definition yields the desired result.
\end{proof}

\begin{remark}\label{rem:explicit_recaled_quantities}
One can notice that an improvement of the above Theorem can be obtained when one uses property (iv) in the definition of $g$ concentration. Indeed, the uniform convergence allows us to find constants $C_i>0$, $i=1,2,3,4$ such that 
$$C_1 H(f|M) \leq \frac{H_N(F_N)}{N} \leq C_2 H(f|M),$$
and
$$C_3 D_\gamma(f|M) \leq \frac{D_{N,\gamma}(F_N)}{N} \leq C_4 D_\gamma(f|M).$$
Moreover, $\br{C_i}_{i=1,\dots,4}$ get closer to $1$ as $N$ increases \emph{in an explicit way}.
\end{remark}

Now that we have established how natural conditioned tensorisation are, we will see that they also give rise to simple example for families that are log-scalable and have the log-power property of order $\beta$.
\begin{theorem}\label{thm:conditoinedlogscale}
Let $f\in P(\R)$ with unit second moment be $g-$concentrated. Then, if there exist $a_1,a_2,C_1,C_2>0$ such that
\begin{equation}\label{uplow}
C_1 e^{-a_1 v^2} \leq f(v) \leq C_2 e^{+a_2v^2}\,
\end{equation}
the family of conditioned tensorisation of $f$ is log-scalable with
$$C_F= \max\left\{\abs{\log C_1},\abs{\log C_2}\right\}+ \max\left\{a_1,a_2\right\}$$
$$+ \sup_{N\in\N}\abs{\frac{\log\pa{2\pa{2\pi}^{-1/2}+2\lambda_N(N)}}{N}-\frac{\log N}{2 N}-\frac{\log\pa{\abs{\mathbb{S}^{N-1}}N^{\frac{N-2}{2}}}}{N}}$$
\end{theorem}
\begin{remark}
Note the sign of the exponent in the right hand side of \eqref{uplow}. The upper bound on $f$ is a very non-restrictive condition. 
\end{remark}
\begin{proof}
As $f$ is $g-$concentrated we have that
$$\frac{\log \pa{\mathcal{Z}_N\pa{f,\sqrt{N}}}}{N}=\frac{\log\pa{2\pa{2\pi}^{-1/2}+2\lambda_N(N)}}{N}-\frac{\log N}{2 N}-\frac{\log\pa{\abs{\mathbb{S}^{N-1}}N^{\frac{N-2}{2}}}}{N}.$$
As the last term converges to $-\pa{1+\log 2\pi}/2$ we conclude that
$$\sup_{N\in \mathbb{N}}\abs{\frac{\log \pa{\mathcal{Z}_N\pa{f,\sqrt{N}}}}{N}} \leq C$$
for some explicit constant $C>0$. Additionally we have that
$$\abs{\frac{\log\pa{f^{\otimes N}(\bm{v})}}{N}} \leq  \frac{\sum_{i=1}^N \abs{\log f(v_i)}}{N} \leq 
\max\left\{\abs{\log C_1},\abs{\log C_2}\right\}+ \max\left\{a_1,a_2\right\} \ .$$
As
$$\frac{\abs{\log F_N\pa{\bm{v}}}}{N} \leq \abs{\frac{\log\pa{f^{\otimes N}(\bm{v})}}{N}} + \abs{\frac{\log \pa{\mathcal{Z}_N\pa{f,\sqrt{N}}}}{N}}$$
the result follows.
\end{proof}

\begin{theorem}\label{thm:conditionedlogpower}
Let $f\in P\pa{\R}$ be $g-$concentrated. Assume in addition that $\|f\|_\infty <\infty$ and that
$$f(v) \geq e^{-\Phi(v)}$$
for a positive $\Phi(v)$. Moreover, assume that 
$$M_{\Phi,\beta} = \int_{\R}\Phi(v)^{1+\beta}f(v)dv < \infty$$
and
$$M_{avg,\Phi,\beta} = \int_{\R^2}\pa{\int_{0}^{2\pi}\Phi(v_1(\theta))^{1+\beta}d\theta}f(v_1)f(v_2)dv_1 dv_2 < \infty.$$
Then the conditioned tensorisation of $f$ has the log-power property of order $\beta$ for $N\geq 3$. Moreover, the constant $C$ in the definition can be evaluated by
\begin{equation}\label{eq:C_of_log_power}
C= \left[ 2^{1+2\beta}\sqrt{3}\frac{\norm{g}_{\infty}+\sup \abs{\lambda_{N-1}}}{\pa{2\pi}^{-1/2}-\sup \abs{\lambda_N}} \pa{2\pa{C_{\epsilon}\norm{f}^{\epsilon}_{\infty}}^{1+\beta}+M_{\Phi,\beta}+M_{avg,\Phi,\beta}}\right]^{\frac{1}{1+\beta}},
\end{equation}
where $C_\epsilon=\sup_{x\geq 1}\frac{\log x}{x^{\epsilon}}$.
\end{theorem}
\begin{proof}
We start by noticing that for $F_N=\frac{f^{\otimes N}}{\mathcal{Z}_N\pa{f,\sqrt{N}}}$ we have that
\begin{equation}\label{eq:reason}
\log\pa{\frac{F_N}{F_N\circ R_{i,j,\theta}}}= \log\pa{\frac{f(v_1)f(v_1)}{{f(v_1(\theta))f(v_2(\theta))}}}.
\end{equation}
As the above only depends on $v_1,v_2$ and $\theta$ we find that, like the proof of Theorem \ref{thm:conditionedproperties} that for any $N\geq 3$
$$\frac{1}{2\pi}\int_{0}^{2\pi}\int_{\kac} \abs{\log \pa{F_N} - \log \pa{F_N\circ R_{1,2,\theta}}}^{1+\alpha}\abs{F_N-F_N\circ R_{1,2,\theta}}d\sigma^N d\theta$$
$$\leq \frac{1}{2\pi}\pa{\frac{N}{N-2}}^{1/2}\frac{\norm{g}_{\infty}+\sup \abs{\lambda_{N-2}}}{\pa{2\pi}^{-1/2}-\sup \abs{\lambda_N}}$$
$$\int_{0}^{2\pi}\int_{v_1^2+v_2^2 \leq N}\abs{\log\pa{f(v_1)f(v_2)}-\log\pa{f(v_1(\theta))f(v_2(\theta))}}^{1+\beta}\abs{f(v_1)f(v_2)-f(v_1(\theta))f(v_2(\theta))}dv_1 dv_2 d\theta$$
$$\leq C_g \Bigg (\int_{0}^{2\pi}\int_{v_1^2+v_2^2 \leq N}\abs{\log\pa{f(v_1)}}^{1+\beta}f(v_2)f(v_1)dv_1 dv_2 d\theta$$
$$+ \int_{0}^{2\pi}\int_{v_1^2+v_2^2 \leq N}\abs{\log\pa{f(v_1(\theta))}}^{1+\beta}f(v_1)f(v_2)dv_1 dv_2 d\theta\Bigg)$$
where $C_g = \frac{2^{1+2\beta}\sqrt{3}}{2\pi}\frac{\norm{g}_{\infty}+\sup \abs{\lambda_{N-1}}}{\pa{2\pi}^{-1/2}-\sup \abs{\lambda_N}}$ and we have used the fact that
$$\pa{a+b}^{q} \leq 2^{q-1}\pa{a^q+b^q}$$
for any $a,b>0$ and $q>1$.\\
Next, we notice that for any non-negative function $h$ with finite $L^\infty$ norm, and for any $\epsilon>0$ we have that
$$\abs{\log h(v)} \leq  C_{\epsilon}h(v)^{\epsilon} \leq C_\epsilon \norm{h}_{\infty}^{\epsilon} $$
when $h(v) \geq 1$. Thus:
\begin{eqnarray}
\int_{0}^{2\pi}\int_{v_1^2+v_2^2 \leq N}\abs{\log f(v_1) }^{1+\beta} f(v_1)f(v_2)dv_1 dv_2 d\theta  &\leq&
 2\pi\pa{C_{\epsilon}^{1+\beta}}\norm{f}_{\infty}^{\epsilon(1+\beta)}
+ 2\pi \int_{\R}\phi(v)^{1+\beta}f(v)dv\nonumber\\
&\leq& 2\pi \left[ \pa{C_{\epsilon}^{1+\beta}}\norm{f}_{\infty}^{\epsilon(1+\beta)}+M_{\Phi,\beta} \right] \ .\nonumber
\end{eqnarray}
Similarly
$$\int_{0}^{2\pi}\int_{v_1^2+v_2^2 \leq N}\abs{\log f(v_1(\theta)) }^{1+\beta} f(v_1)f(v_2)dv_1 dv_2 d\theta  \leq 2\pi 
\left[ \pa{C_{\epsilon}^{1+\beta}}\norm{f}_{\infty}^{\epsilon(1+\beta)}+M_{avg,\Phi,\beta} \right],$$
completing the proof.
\end{proof}
\begin{remark}\label{rem:importanceoflogpowerdefintoin}
We would like to emphasise that \eqref{eq:reason} is {\em{exactly}} what motivated the definition of the log-power property and why Theorem \ref{thm:mainlofpower} works. The problem with the Kac Walk lies mainly with the strong dependency in the dimension $N$. The log-power property states, in some sense that $\abs{\log \pa{\frac{F_N}{F_N \circ R_{1,2,\theta}}}}^{1+\beta}$ has {\em{lost most f its variables}} and so the remaining $\abs{F_N-F_N\circ R_{1,2,\theta}}$ poses no problems due to the additional integration over $\kac$.
\end{remark}

The last thing we will do in this section, is to set the final piece that will be needed to use the mean field limit approach to move from Theorem \ref{thm:mainlofpower} to Theorem \ref{thm:almost_cerc_kac_boltz}. The main idea is to somehow have both the log-power property, and be strongly entropically chaotic.

\begin{definition}\label{def:log-power-chaoticity}
We say that a family of probability densities $\br{F_N}_{N\in\mathbb{N}}\in P\pa{\kac}$ are $f-$log-power-chaotic of order $\beta>0$ if $F_N$ is $f-$strongly entropic and has the log-power of order $\beta$ property.
\end{definition}

\begin{prop}\label{prop:logchaoticexample}
Let $f\in P\pa{\R}$ be $g-$concentrated. Assume in addition that $I(f)<\infty$ and that there exists $k\geq 2$, $\beta>0$ such that $M_{k(1+\beta)}(f) <\infty$ and
$$f(v) \geq e^{-a\abs{v}^k}.$$
Then the conditioned tensorisation of $f$ is log-power-chaotic of order $\beta$ when $N\geq 3$. Moreover, the constant $C$ in the definition can be evaluated by
\begin{equation}\label{eq:logchaoticexample}
C= \left[ 2^{1+2\beta}\sqrt{3}\frac{\norm{g}_{\infty}+\sup \abs{\lambda_{N-1}}}{\pa{2\pi}^{-1/2}-\sup \abs{\lambda_N}} \pa{2\pa{C_{\epsilon}I(f)^{\frac{\epsilon}{2}}}^{1+\beta}+\pa{1+C_{k,\beta}}M_{(k(1+\beta)}}\right]^{\frac{1}{1+\beta}},
\end{equation}
where $C_\epsilon=\sup_{x\geq 1}\frac{\log x}{x^{\epsilon}}$, $C_{k,\beta}= 2^{k(1+\beta)+1}\pa{\int_{0}^{2\pi} \abs{\cos\theta}^{k(1+\beta)}d\theta}$.
\end{prop}
\begin{proof}
This follows immediately from Theorems \ref{thm:conditionedproperties} and \ref{thm:conditionedlogpower}. In the notations of the latter we have that $\phi(v)=\abs{v}^k$ and hence
$M_{\phi,\beta}=M_{k(1+\beta)}(f)$,
and it is easy to show that
$$M_{avg,k,\beta} \leq 2^{k(1+\beta)+1}\pa{\int_{0}^{2\pi} \abs{\cos\theta}^{k(1+\beta)}d\theta}M_{k(1+\beta)}.$$
Also, as $\norm{f}_{\infty} \leq \sqrt{I(f)}$, we obtain the desired result from \eqref{eq:C_of_log_power}.
\end{proof}

\section{The Validation of Kac's Program}\label{sec:vision}
Gathering tools developed in  Section \ref{sec:functionalineq} and \ref{sec:g concentration}, we will now prove Theorem \ref{thm:almost_cerc_kac_boltz}.
\begin{proof}[Proof of Theorem \ref{thm:almost_cerc_kac_boltz}]  As noted above, $\norm{f}_{\infty} \leq \sqrt{I(f)}$, 
and then since $L^p \subset L^1\cap L^\infty$ for any $p\geq 1$, the bound on $I(f)$ implies a bound on all $L^p$ norms. 
Using Theorems \ref{thm:mainlofpower}, \ref{thm: CCLLV approximation thm},
the bound on $I(f)$ and Proposition \ref{prop:logchaoticexample}, we conclude the inequality \eqref{eq:almost_cerc_kac_boltz}. Since we are taking the limit of $N$ to infinity (and this can use convergence of moments for large $N$ and the disappearance of $\lambda_N$), we can use the constant
\begin{equation}\nonumber
\mathcal{C} = \frac{k\beta-(1+\beta)}{(1+\beta)(1-\gamma)}\pa{\frac{(1+\beta)(1-\gamma)}{k\beta-(1+\beta)}}^{\frac{k\beta-\gamma(1+\beta)}{k\beta-(1+\beta)}}
\frac{2^{\frac{1-\gamma}{k\beta-(1+\beta)}}}{3^{\frac{2k\beta-k\beta\gamma-\gamma(1+\beta)}{k\beta-(1+\beta)}}}
\frac{1}{C_f^{\frac{(1+\beta)(1-\gamma)}{k\beta-(1+\beta)}}\pa{1+2M_{2k}(f)}^{\frac{\beta(1-\gamma)}{k\beta-(1+\beta)}}}\ .
\end{equation}
where $C_f=\left[ 2^{1+2\beta} \sqrt{3} \pa{2\pa{C_{\epsilon}I(f)^{\frac{\epsilon}{2}}}^{1+\beta}+\pa{1+C_{k,\beta}}M_{(k(1+\beta)}(f)}\right]^{\frac{1}{1+\beta}}$ with $\epsilon>0$ arbitrary and $C_\epsilon=\sup_{x\geq 1}\frac{\log x}{x^{\epsilon}}$, $C_{k,\beta}= 2^{k(1+\beta)+1}\pa{\int_{0}^{2\pi} \abs{\cos\theta}^{k(1+\beta)}d\theta}$. A simpler argument shows \eqref{eq:almost_cerc_kac_boltz_exp}.\\
The second statement of the theorem follows immediately from the fact that $I(f)$ is monotone decreasing along solutions of the Kac-Boltzmann equation \cite{McK} and all moments of the solution to the Kac-Boltzmann equation that are finite initially  remain uniformly bounded in time \cite{Des}. This allows us to replace $\mathcal{C}(f(\cdot,t))$ above, with a uniform constant that only depends on the initial datum.
\end{proof}
Before we move to our final result, we would like to emphasise that this connection between inequalities for the Kac Walk and the Kac-Boltzmann equation is not a coincidence. Conditioned tensorisation of a given function $f$ allows us to move back and forth between the equations - the result of which is Theorem \ref{thm:almost_cerc_kac_boltz}. This is expressed in the following Theorem, stating that the optimality of the 'almost' Cercignani's conjecture is equivalent (in some sense) in these two settings.

\begin{theorem}\label{thm:alost cerc on sphere is optimal}
Assume that $f\in P\pa{\R}$ is $g-$concentrated and satisfies \eqref{Cercon2} for some $\epsilon>0$. Then the conditioned tensorisation of $f$ satisfies
$$\frac{D_{N,\gamma}(F_N)}{N} \geq C_1 \pa{\frac{H_N(F_N)}{N}}^{1+\epsilon}, $$
for an appropriate constant $C_1$. In particular, if \eqref{Cercon2} is optimal, then so is its Kac's Walk counterpart. 
\end{theorem}
\begin{proof}
Due to Remark \ref{rem:explicit_recaled_quantities} we can find constants such that
$$\frac{H_N(F_N)}{N} \leq C_1 H(f|M),\quad \frac{D_{N,\gamma}(F_N)}{N}\geq C_2 D_0(f).$$
As 
${\displaystyle \Q_\gamma(f) \geq K_1 H(f|M)^{1+\epsilon}}$
with an appropriate constant $K_1$, depending on $f$, we have that
$$\frac{D_{N,\gamma}(F_N)}{N} \geq K_1 C_2\pa{\frac{H_N(F_N)}{C_1 N}}^{1+\epsilon},$$
completing the proof of the first statement. The latter statement follows from the above and Theorem \ref{thm:conditionedproperties}.
\end{proof}

\end{document}